\documentclass{article}

\usepackage{PRIMEarxiv}

\usepackage[utf8]{inputenc} % allow utf-8 input
\usepackage[T1]{fontenc}    % use 8-bit T1 fonts
\usepackage{hyperref}       % hyperlinks
\usepackage{url}            % simple URL typesetting
\usepackage{booktabs}       % professional-quality tables
\usepackage{amsfonts}       % blackboard math symbols
\usepackage{nicefrac}       % compact symbols for 1/2, etc.
\usepackage{microtype}      % microtypography
\usepackage{lipsum}
\usepackage{fancyhdr}       % header
\usepackage{graphicx}       % graphics

\usepackage[export]{adjustbox}

\usepackage[english]{babel}

\usepackage{doi}

\usepackage{comment}  % to use comment

\usepackage{amsmath} % math 

\usepackage{amssymb} % math symbol

\usepackage{amsthm}

\usepackage[ruled,linesnumbered]{algorithm2e}%ruled, vlined

\usepackage{float}

\usepackage{subcaption}

\usepackage{silence}

\usepackage{dirtytalk}

\newtheorem{theorem}{Theorem}[section]
\newtheorem{lemma}[theorem]{Lemma}
\newtheorem{definition}[theorem]{definition}

\graphicspath{{media/}}     % organize your images and other figures under media/ folder

\title{An approach to optimizing the VABA protocol using $\kappa$-size committee}
\author{ {Nasit S Sony} \\
	University of California, Merced\\
	CA 95340, USA \\
	\texttt{nsony@ucmerced.edu} \\
}

\begin{document}
\maketitle

\begin{abstract}

Byzantine agreement protocols in asynchronous networks have gained renewed attention due to their independence from network timing assumptions to ensure termination. Traditional asynchronous Byzantine agreement protocols require every party to broadcast its requests (e.g., transactions), leading to high communication costs as parties ultimately agree on one party's request. This inefficiency is particularly significant in multi-valued Byzantine agreement protocols, where parties aim to agree on one party's requests under the assumption $n=3f+1$, where $n$ is the total number of parties, and $f$ is the number of Byzantine parties.

To address these inefficiencies, we propose Efficient-VABA (eVABA), an optimized protocol for the asynchronous Byzantine agreement (ABA) problem. By limiting broadcasts to a selected subset of parties, the protocol reduces the number of messages and computation overhead.

\end{abstract}

% keywords can be removed
\keywords{ Blockchain, Distributed Systems, Byzantine Agreement, System Security}

\section{Introduction}

 Byzantine agreement (BA) is a fundamental problem in computer systems introduced by Lamport, Pease, and Shostak in their pioneering works \cite{BYZ22,BYZ23}. BA assumes a system where multiple computers try to agree on a value, but some of the computers might be byzantine (whose behavior is arbitrary, unpredictable, and malicious). We will use machine, party, and node synonymously. Many models with numerous system assumptions have been presented since the seminal work \cite{BYZ23} to solve the BA problem.

 Bitcoin \cite{BITCOIN01} has refueled the interest in the study of the Byzantine agreement problem. From the literature, we see the discussions of the three widely used models: 1) an asynchronous model, where every party broadcasts its requests; 2) a partially synchronous model, where the network becomes synchronous after an asynchronous period, and a party is responsible for handling the client requests; and 3) a synchronous model. But, the synchronous or partially synchronous protocols perform well when the message delivery is guaranteed in a time bound, and the number of parties is small. The problem with these protocols is that they run around a leader, and the leader can be faulty. Though there is a mechanism that can change a faulty leader and elect a new leader, Miller et al. \cite{HONEYBADGER01} prove that a protocol in a partially synchronous model is unable to make progress in \textit{intermittently synchronous network}.

%\textbf{Research question} Does having multiple leader reduce complexity of reaching consensus?

%\paragraph{Partially synchronous to asynchronous model.} Since the protocols that assume synchronous or partially synchronous models are not immune to \textit{intermittently synchronous network}, we must think of the deployment of the asynchronous protocols. But, most of the asynchronous BA protocols are theoretical and inefficient \cite{ SECURE02,SECURE03,SECURE06}. VABA \cite{BYZ17} is an $O(n^2)$ algorithm with the running time $O(1)$ in expectation where parties agree on one party's requests. 
The byzantine agreement works in the presence of byzantine faults, and asynchronous protocols always let every party propose its requests. Miller et al. \cite{HONEYBADGER01} provides an atomic broadcast protocol that outputs multiple parties' requests at a time, which is followed by \cite{FASTERDUMBO,PMVBA, OHBBFT,cMVBA,SlimABC}. The atomic broadcast protocol lets the parties agree on multiple parties' requests and incurs a high communication complexity. But if different parties broadcast the requests and have duplication, then allowing every party proposes each time is not beneficial in terms of message and communication complexity. Therefore, we take the question :

\textit{Is it possible to reach an agreement without allowing every party to broadcast its requests? Does it provide any improvement over the existing protocols? }

\subsection{Challenges and Proposed Solutions}
To design a protocol that reduces the number of messages and computational overhead, we address the following key challenges:
\begin{enumerate}
    \item \textbf{Determining Necessary Broadcasts:} How many broadcasts are essential to maintain protocol progress?
    \item \textbf{Selection of Broadcasting Parties:} How can we select a subset of parties for broadcasting?
    \item \textbf{Preventing Dishonest Broadcasts:} How can we prevent non-selected dishonest parties from broadcasting?
    \item \textbf{Dispersing Broadcasts:} How can the broadcast messages efficiently reach all parties?
    \item \textbf{Adapting Leader Election:} How can we adapt the leader election process to ensure leaders are chosen from the selected subset?
\end{enumerate}

To address the above-mentioned challenges, we have done the following analysis and propose a solution based on the analysis.

\begin{itemize}
    \item \textbf{Analysis of Adversarial Behavior:} A dishonest party or adversary can read and delay messages from honest parties but must eventually deliver them. Therefore, to ensure protocol progress, we employ a committee that ensures the inclusion of an honest party.
    \item \textbf{Enhancements to the Provable-Broadcast Protocol:}. To stop a dishonest party from broadcasting its proposal, we introduce a security check in the provable-broadcast protocol. Honest parties only respond to proposals from selected parties, preventing dishonest parties from completing the first step of the proposal-promotion sub-protocol.
    \item \textbf{Efficient Broadcast Handling:} In asynchronous protocols, the standard practice is to wait for $n-f$ broadcasts to ensure progress, given that $f$ parties may be faulty. However, if only one honest party broadcasts, other parties must respond upon receiving the first broadcast to maintain liveness. We adapt this behavior to align with the reduced broadcasting subset.
    \item \textbf{Adapting Leader Election for Selected Parties:}The standard leader election protocol is insufficient as it may elect any party, including non-selected ones, as a leader. To address this, we combine the leader election protocol with a mapping mechanism that ensures the elected leader belongs to the selected committee. This adaptation guarantees that the leader has been chosen from the set of vetted parties, maintaining the integrity of the protocol.
\end{itemize}

\subsection{Our Contribution}
We design a protocol that improves the efficiency of the existing protocols. Our two main observations are that we can reduce the number of proposals, and the reduction helps to design an efficient protocol. We apply our reduction technique to design a validated asynchronous byzantine agreement (VABA) protocol.

\begin{itemize}
    \item Integrate the committee selection protocol that ensures the inclusion of at least one honest party with high probability, enhancing resilience against Byzantine failures.
    \item Propose eVABA, reducing communication rounds by selectively broadcasting proposals.
    \item Integrate a leader election protocol and introduce the mapping of the leader to a committee member to achieve an agreement on a proposal by parties.
    \item Provide theoretical analysis showing the effectiveness and the correctness of the protocol.
    % \item Validate cMVBA’s efficiency and scalability through case studies, demonstrating superior performance over traditional MVBA protocols.
\end{itemize}

\section{The Design of eVABA}
%\subsection{Efficient-VABA} 
The objective of our design is to minimize the number of messages and the related computations in a VABA execution. The asynchronous byzantine agreement protocols let every party propose their requests (a request is a value $v$), and at the end of the protocol, parties agree on one party's proposed request. As illustrated in the Introduction, if the VABA protocol lets $n$ parties broadcast their requests, then the parties agree on a completed broadcast (decide on a broadcast) with $\frac{2}{3}$ probability. However, we also observe that if parties are unable to decide in one view, then the parties adopt the elected parties' requests in the next view. Therefore, if we consider the probability of a completed broadcast in the first view, then we can select a subset of parties as a proposer in each view and minimize the number of messages and computations. The proposed protocol contains five components, each described in the subsections below. Figure \ref{fig:E-VABA} provides an overview of the protocol.
\begin{figure}[ht]
    \centering
    \includegraphics[width=0.8352\textwidth,height=0.2376\textwidth]{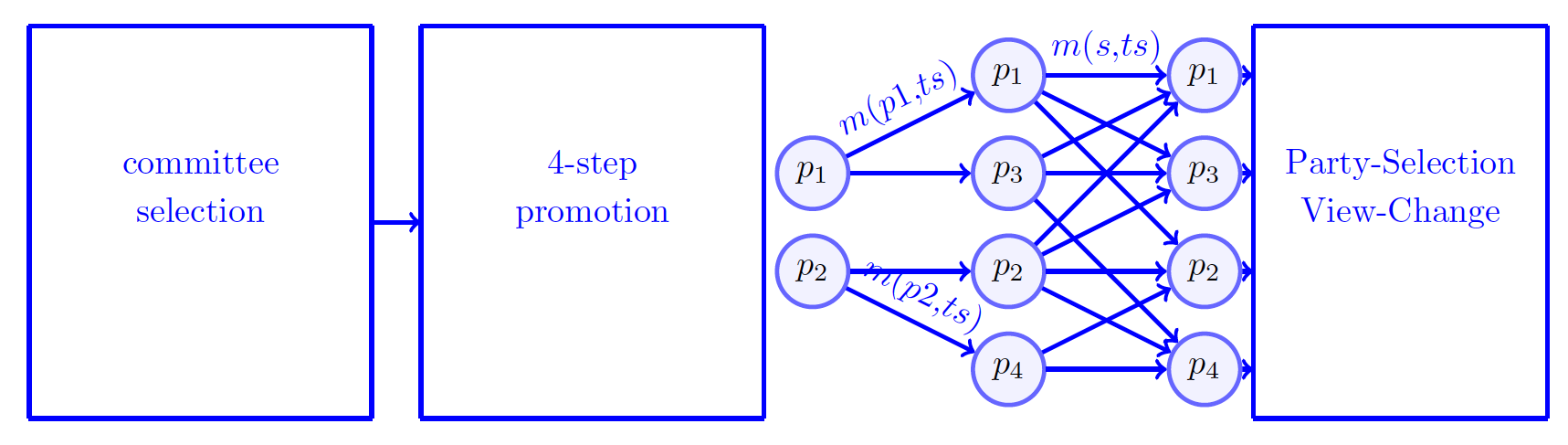}
    \caption{eVABA illustration. Parties select a committee, and each committee member promotes its requests using the 4-step promotion. Then, a party proposes the proposal $p_i$ with proof $ts$. After receiving a proposal, each party suggests the proposal and waits for $\langle n-f \rangle$ suggestions. Then, the selected parties move to the party-selection and view-change step.}
    \label{fig:E-VABA}
\end{figure}

\subsection{Committee Selection}
To ensure the progress of the protocol, we must select a subset of parties in a way that guarantees at least one honest party within the group. We employ a standard committee selection protocol, setting $\kappa$ as the security parameter, bounded by $n$. The selection process is randomized, ensuring that at least one selected party is honest with overwhelming probability.

\subsection{Proposal-Promotion}
The eVABA protocol works in a view-based manner, and in each view, at least one party is required to complete its 4-step broadcast and collect the related proof. Therefore, each selected party broadcasts its proposal in 4-steps. Each step contains a two-step broadcast that generates proof for that step. Therefore, we call this $provable broadcast$. Each step of the proposal promotion contains a provable broadcast, and the deliverables of steps are $key$, $lock$, and $commit$. Since the eVABA allows only the committee members to broadcast their requests, we adopted a modified version of the $provable-broadcast$ and named it prioritized provable-broadcast. We provided the properties and details about the protocol below.

\paragraph{Prioritized probable broadcast (P-PB)\label{P-PB-D}} 

While $\kappa$ parties are chosen as promoters, a non-selected Byzantine party may attempt to promote its proposal and disrupt the network by consuming bandwidth. As a result, the standard $provable-broadcast$ sub-protocol cannot be used directly for promoting requests. To address this, we introduce a modified version of $provable-broadcast$, called the Prioritized-Provable Broadcast (P-PB). In this modified protocol, a party only responds to broadcasts from selected committee members, effectively preventing Byzantine parties from promoting unauthorized requests. The construction and components of P-PB are detailed in Algorithm \ref{PB:P-PB1}.
$provable-broadcast$ for the selected parties can be obtained through the same threshold-signature scheme. Generally, in $provable-broadcast$, a party adds sign-share to a request if the request is authentic or sent by some party. But we design a protocol where a party adds sign-share to a request only if the request is from a selected party for the particular view. Therefore, only the selected parties receive $\langle n-f \rangle$ replies with $sign-share$ and produce a $threshold-signature$ message. 

\subparagraph*{Construction of Prioritized-PB:} The pseudocode of the P-PB protocol is given in Algorithm \ref{PB:P-PB1}. We provide a flow of the protocol below: 

\begin{itemize}
    \item Upon the invocation of a P-PB protocol, a party  creates a message type $SEND$ using the $ID$ and the received value and proof $\langle v, \sigma \rangle$. Then the party multi-casts/broadcasts the created message. (lines 03-06)
    \item Upon receiving a $SEND$ type message from a party $p_k$, a party checks whether the sender is a selected party for the particular $view$. Then the party checks the value and proof. The checking depends on the view and the step of the P-PB.
    \item  Upon receiving a $sign-share$ $\sigma_k$ from a party $p_k$, a party checks the authenticity of the message. If the sender is authentic, the party adds the $sign-share$ $\rho$ to its set $\Sigma$. (line 08-09)
    \item A selected sender waits for $\langle n-f \rangle$ valid $sign-shares$ and upon receiving the valid $sign-shares$, the party returns the $threshold-signature$ of the $sign-shares$. (lines 05-06)
\end{itemize}

\begin{algorithm}[hbt!]
%\SetAlgoLined
\LinesNumbered
\DontPrintSemicolon
\SetAlgoNoEnd
\SetAlgoNoLine

\SetKwProg{LV}{Local variables initialization:}{}{}
\LV{}{
   $\Sigma \leftarrow \{\}$\;
}

\SetKwProg{un}{upon}{ do}{}
\un{$P-PB \langle ID,\langle v, \sigma \rangle\rangle $ invocation}
{
  $multi-cast\langle ID, SEND, \langle v, \sigma \rangle \rangle$ \;
  %$msg$ $\leftarrow$ $\langle id, requests \rangle$ \;
  %$threshold-sign\langle msg \rangle$\;
  %$multi-cast \langle msg \rangle$\;
  \textbf{wait} until $|\Sigma| = n-f$\;
  
  \KwRet $threshold-sign\langle \Sigma \rangle$\; %$CombineShares \langle id, \Sigma \rangle$\;
}
\SetKwProg{un}{upon }{ do}{}

\un{receiving $\langle ID, ACK, \sigma_k \rangle$ from a party $p_{k}$ for the first time}
{
 \uIf{$share-validate (\langle ID, ACK \rangle, k, \sigma_k ) = true$ }{
    $\Sigma \leftarrow {\sigma_k \cup \Sigma}$}
}

\caption{P-PB: Protocol for party  $p_i$}\label{PB:P-PB1}
\end{algorithm}

\subsection{Propose-Suggest}
In eVABA, the Propose-Suggest step ensures the proper dissemination of provable-broadcast messages from committee members. Since only committee members are allowed to broadcast their requests and gather threshold signatures, a party cannot wait for $n-f$ threshold signatures upon receiving a proposal. However, it is essential that $n-f$ commit proofs are distributed to and received by at least $n-f$ parties. During the Suggest step, every party ensures it has at least one commit proof and broadcasts this proof. Consequently, a party waits to receive $n-f$ commit proofs before proceeding to the next step of the protocol.
\subsection{Party Selection}
In the party-selection step, honest parties coordinate to elect a leader using the standard leader election protocol. However, the standard leader election protocol selects a leader from the entire set of $n$ parties, while the proposed protocol requirement is to choose a leader exclusively from the committee members. To achieve this, each party first executes the leader election protocol, which may select any party from ${1, 2, ..., n}$. If the elected leader is a committee member, the party proceeds to deliver that leader's proposal through the view-change step. If the elected leader is not a committee member, the party employs a mapping function to map the elected leader to one of the committee members. The mapping function ensures consistency across all parties, such that the same committee member is selected regardless of the party performing the mapping. Once a valid committee member is identified, the parties deliver that member's proposal through the view-change step.
\paragraph{Mapping to a selected party} The leader election protocol ensures that all parties select the same leader. The leader is identified by a unique $id$. To map the leader to the committee, each party calculates the distance between the elected leader and the set of selected committee members. By choosing the committee member with the smallest distance to the elected leader (the difference of the $ids$, the protocol maintains the core property of leader election (same leader for each party) while ensuring the selected leader belongs to the committee.

\subsection{View-Change}
The purpose of the view-change step is to disseminate the selected party's delivery (the outcome of the chosen leader's proposal-promotion) to every party. The proposal-promotion guarantees that if a committee member delivers a commit, then at least $n-2f$ honest parties deliver lock. Therefore, all honest parties receive this lock in the view-change exchange. Similarly, if an honest party delivers a lock, then $f+1$ honest parties receive this key. A party waits for $n-f$ view-change messages; thus, it can get at least one honest party's delivery among $f+1$ honest parties. The decision is straightforward forward: if a party receives a commit message, then the party decides on the value. If a party receives a lock, it increases its variable to indicate its latest key. Last, if it receives a key, it updates its KEY variable to store the current view and the received key. If a party can not reach a decision, then the party adopts the value v of its key and moves to the next view, where it promotes the updated value together with the proof for the external validation function (EX-PB-VAL) of the Provable-Broadcast.

\section{Security and Efficiency Analysis}
\subsection{Security Analysis}
\paragraph{Prioritized provable-broadcast} The proposed P-PB protocol satisfies the P-PB-integrity, P-PB-validity, P-PB-abandonability and P-PB linear complexity properties from the code. (similar to provable-broadcast)
\newline 
The P-PB protocol satisfies the termination and provability properties from the provable-broadcast.

\begin{lemma} \label{E-VABA: selected}
    Algorithm \ref{P-PB-D} satisfies the selected property.
\end{lemma}

\begin{proof}
   From the provability property of the protocol, a party generates a $threshold-signature$ proof for its broadcast. A party requires at least $ f+1 $ honest parties' $sign-shares$ on the promoted message to generate a $threshold-signature$. From the Algorithm, an honest party replies with a $sign-share$ only if the sender is selected for the view. Therefore, only the selected parties can complete the promotion.
\end{proof}

%\paragraph{Efficient-VABA}

\begin{theorem}
   The eVABA protocol satisfies the Agreement, External-Validity, and Liveness/Termination properties of the VABA protocol, given that the underlying Committee selection, provable-broadcast, and the leader-election protocols are secure. 
\end{theorem}

\subsection{Efficiency Analysis}

First, let us briefly go through the steps of the eVABA protocol. From the protocol overview, the message exchange happens when the parties select $\kappa$ parties using the committee selection protocol, the selected parties promote their request using the $proposal-promotion$ sub-protocol, propose and suggest a completed promotion using the $propose-suggest$ step, elect a leader, and deliver the value and proof using view-change.

The proposed protocol ensures the termination in the view change step. Since the protocol exits after the constant rounds of messages, the expected \textit{time complexity} of the protocol is $O(1)$. The main message and communication cost of the eVABA protocol comes from the four-step $proposal-promotion$ sub-protocol invocations. In each step of the $proposal-promotion$ sub-protocol, in the first round, $\kappa$ parties broadcast the messages to $n$ parties, and in the second round, $n$ parties reply to $\kappa$ parties. So the total number of messages are $n*\kappa$, which reduces the total number of messages and the related computations from $n^2$ $n\kappa$. The $suggestion$, $leader-election$ and $view-change$ steps exchange three messages, and the communication is all-to-all, so the message complexity is $O(n^2)$ and the communication complexity is $O(n^2(L+K))$. Both are optimal. Here, $L$ is the bit length of the message and the proof; $K$ is the bit length of digital signatures.

\section{Conclusion}

This paper introduced an optimized VABA protocol that minimizes the number of communication messages between the parties via a committee-based approach. By dynamically selecting $\kappa$ parties to broadcast, the protocol maintains optimal resilience and achieves optimal message complexity of \(O(n^2)\) and expected constant asynchronous rounds. The introduction of a committee to reduce the number of broadcasts further enhances the effectiveness of the protocol, making the protocol ideal for fault-tolerant state machine replication protocols. Future work will explore protocol optimizations and empirical studies to prove the protocol's robustness in diverse settings.

%We presented two asynchronous byzantine agreement protocols. The core technical contributions include the introduction of the prioritized provable-broadcast (P-PB) protocol. The P-PB protocol broadcasts only the selected parties' requests and provides an $O(n|v|)$ communication complexity for a value $v$ compared to $O(n|v| + Kn^2logn)$ communication complexity of the RBC protocol. Though the P-PB protocol is unable to provide the totality property of the RBC protocol, we overcome the limitations by adding an extra round of messages.

%We presented Efficient-VABA, an instantiation of VABA \cite{BYZ17} that has reduced the number of messages and the related computations from $n^2$ to $(f+1) n$. We also presented an atomic broadcast protocol, Slim-HBBFT, that provides an instantiation of the HoneyBadgerBFT\cite{HONEYBADGER01} protocol that performs better when parties have duplicate requests. The proposed protocol aimed to agree on the fraction of $n$ parties' requests and removed the $kn^3logn$ term from the communication complexity. 

%\bibliographystyle{ACM-Reference-Format}
\bibliography{references}

\appendix

\end{document}